\documentclass[11pt]{article}

\usepackage{fullpage}
\usepackage{graphicx}
\usepackage{amsthm} 
\usepackage{amsmath} 
\usepackage{amsfonts} 

\usepackage{amssymb}

\usepackage{xspace}   
\usepackage{adjustbox} 
\usepackage[shortlabels]{enumitem} 

\usepackage{tikz}
\usetikzlibrary{shapes}
\usepackage{forest}

\newcommand{\lineNumberPatch}{}

\usepackage{listings}
\usepackage[hidelinks]{hyperref}

\newtheorem{theorem}{Theorem}
\newtheorem{corollary}{Corollary}

\newtheorem{lemma}{Lemma}

\newcommand{\margincomment}[1]%
{{\marginpar{{\footnotesize\begin{minipage}{\marginparwidth}%
          \begin{flushleft}%
            {#1}%
          \end{flushleft}%
        \end{minipage}%
      }}}\ignorespaces}

\newcommand{\nealrevisedtext}[1]{{\color{darkblue}#1}}

\newcommand{\ignore}[1]{}

\colorlet{darkgreen}{green!45!black}
\colorlet{darkblue}{blue!45!black}

\makeatletter
\newcommand*{\etal}{%
    \@ifnextchar{.}%
       {et al}%
      {et al.\@\xspace}%
}
\makeatother
\newcommand{\etals}{et al.'s\xspace}

\makeatletter
\newcommand{\mycase}[1]{
  \medskip\par\noindent\underline{Case #1}: 
  \@ifnextchar\par\@gobble\relax
}
\makeatother

\newcommand{\braced}[1]{{ \left\{ #1 \right\} }}
\newcommand{\angled}[1]{{ \left\langle #1 \right\rangle }}

\renewcommand{\cal}{\mathcal}

\newcommand{\Queries}{{\cal S}}

\newcommand{\Keys}{{\cal K}}

\newcommand{\depth}{\ensuremath{\operatorname{\sf depth}}}

\newcommand{\cost}{{\textsf{cost}}}
\newcommand{\optcost}{{\textsf{opt}}} 
\newcommand{\query}{{q}}
\newcommand{\compnode}[2]{\ensuremath{\angled{\query{{}\mathbin{#1}{}}{#2}}}}

\newcommand{\leaves}[1]{\ensuremath{{\sf leaves}(#1)}}
\newcommand{\weight}[1]{{w_{#1}}}  
\newcommand{\nodeweight}[1]{{w_{#1}}}
\newcommand{\treeweight}[1]{{w(#1)}}
\newcommand{\sideweight}[1]{{\textsf{sw}({#1})}} 
\newcommand{\sw}[1]{\sideweight{#1}}

\newcommand{\key}[1]{k_{#1}}
\newcommand{\querysubset}[3]{S(#1, #2, #3)}

\newcommand{\treeLabelSize}{\scriptsize}

\tikzset{
  edgeYes/.style = {left, node contents={\scriptsize \emph{y}}, pos=0.3}, 
  edgeNo/.style = {right, node contents={\scriptsize \emph{n}}, pos=0.26}, 
  edgeLt/.style = {left, node contents={\scriptsize \emph{$<$}}, pos=0.35},
  edgeEq/.style = {left, node contents={\scriptsize \emph{$=$}}, pos=0.37, inner sep=0},
  edgeGt/.style = {right, node contents={\scriptsize \emph{$>$}}, pos=0.35},
  weight/.style = {label={270:\treeLabelSize #1}}, 
  every path/.style = {draw=black!40}, 
}

\newenvironment{TREE}{
  \forest
    for tree = { draw=black!40, s sep=0.6em, l sep=1.3em }, 
    where n children = 0 {
      rectangle,                
      minimum width=14.5pt,
      inner sep=2pt, 
      text height=7.5pt,
      text depth=1.5pt,
    }{
      rounded rectangle,        
      sibling distance=2em,
      text height=7pt,
      text depth=2.2pt,
    },
    where n={2}{where n'={1}{edge label={node[edgeNo] }}{}}{}, 
    where n={1}{where n'={2}{edge label={node[edgeYes] }}{}}{}, 
  }{%
  \endforest
}

\newenvironment{ThreeWayTree}{
  \forest
    for tree = { draw=black!40, s sep=1.5em, l sep=1.3em }, 
    where n children = 0 {
      rectangle,                
      minimum width=14.5pt,
      inner sep=2pt, 
      text height=7.5pt,
      text depth=1.5pt,
    }{
      rounded rectangle,        
      sibling distance=2em,
      text height=7pt,
      text depth=2.2pt,
    },
    where n={3}{where n'={1}{edge label={node[edgeGt] }}{}}{},
    where n={2}{where n'={2}{edge label={node[edgeEq] }}{}}{},
    where n={1}{where n'={3}{edge label={node[edgeLt] }}{}}{},
  }{%
  \endforest
}

\newcommand{\subtreecolor}{black!4}
\newcommand{\interiorcolor}{green!4}
\newcommand{\edgecolor}{black!66}

\newcommand{\subtreeparams}[2]{
  \forestset{
    subtree1/.style={
      anchor=north, 
      child anchor=north, 
      outer sep=0pt,
      yshift=10pt,
      rounded corners=0.5em, 
      shape=semicircle,
      inner sep=2pt,
      text depth=4pt,
      xscale=#1,
      content format={
        \noexpand\scalebox{#2}[1]{\forestoption{content}}
      },
    },
  }
  \forestset{
    subtree/.style={
      subtree1,
      draw={\edgecolor, thin, dotted}, 
      edge={\edgecolor, thin}, 
      fill=\subtreecolor, 
    },
  }
  \forestset{
    dimtree/.style={
      subtree1,
      draw={gray, thin, dotted},
      edge={gray, thin, dotted}, 
      fill=gray!3,
      text=black!50, 
    },
  }
  \forestset{
    greentree/.style={
      subtree1,
      draw={\edgecolor, thin}, 
      edge={\edgecolor, thin}, 
      fill=\interiorcolor, 
    },
  }
}
\newcommand{\defaultsubtreeparams}{
  \subtreeparams{1}{1}
}
\defaultsubtreeparams 

\title{A Simple Algorithm for Optimal Search Trees \\ with Two-Way Comparisons
\thanks{
  To appear (without Appendix~\ref{sec: root stronger}) in ACM Transactions on Algorithms~\cite
  {chrobak_etal_simple_bcst_algorithm_2021}. 
  See~\cite[Theorem 1]{chrobak_etal_isaac_2015} for a conference version of the result in this paper
  (but the algorithm in that version is more complicated, and its analysis is less informative).
  See~\cite{DBLP:journals/corr/ChrobakGMY15,chrobak_etal_huangs_algorithm_2018,CHROBAK2021104707}
  for improved versions of other results in~\cite{chrobak_etal_isaac_2015}.
}
}

\author{
  Marek Chrobak\thanks
  {
    University of California at Riverside. 
    Research supported by NSF grants CCF-1217314 and CCF-1536026.
  }
  \and
  Mordecai Golin\thanks
  {
    Hong Kong University of Science and Technology.
    Research funded by HKUST/RGC grant FSGRF14EG28 and RGC CERG Grants 16208415 and 16213318.
  }
  \and
  J. Ian Munro\thanks
  {
    University of Waterloo.
    Research funded by NSERC and the Canada Research Chairs Programme.
  }
  \and
  Neal E. Young\thanks
  {
    University of California at Riverside.
    Research supported by NSF grant IIS-1619463.
  }
}

\begin{document}


\maketitle

\begin{abstract} 
We present a simple $O(n^4)$-time algorithm for computing optimal search trees with two-way comparisons.
The only previous solution to this problem, by Anderson \etal, has the same running time, but
is significantly more complicated and
is restricted to the variant where only successful queries are allowed.
Our algorithm  extends directly to solve the standard full variant of the problem, which also allows 
unsuccessful queries and for which no polynomial-time algorithm was previously known. The
correctness proof of our algorithm relies on a new structural theorem for two-way-comparison search trees.
\end{abstract}


\thispagestyle{empty}

\setcounter{page}{1}


\section{Introduction}\label{sec: introduction}



Search trees are among the most fundamental data structures for storing and accessing data.
\emph{Static} search trees
are used in applications where the set of possible queries and their frequencies are known in advance,
in which case a single tree can be precomputed and used throughout the application to handle lookup queries.
Such a tree is \emph{optimal} if it minimizes the expected search cost.
The problem of computing optimal search trees has been studied extensively in various forms since the 1960s.
Perhaps the most famous example is Knuth's $O(n^2)$-time algorithm~\cite{Knuth1971},
which is widely considered one of the gems of algorithmics
and has inspired the discovery of general techniques
for speeding up dynamic-programming algorithms~\cite[\S7.1]{Bein2013}.

But most optimal-search-tree results, including Knuth's, consider only trees based on three-way
comparisons~\cite{Knuth1971}.
In contrast, standard high-level programming languages implement two-way comparisons (e.g., ${=}, {\le}, {<}$).
As pointed out by Knuth himself in the second edition of \emph{``The Art of Computer Programming''},
it is desirable to understand search trees that use only two-way comparisons~\cite[\S6.2.2 ex.~33]{Knuth1998}.
Yet they are still not well understood: there is only one published proof that optimal two-way-comparison trees
can be computed in polynomial time,
given decades after Knuth's result,
in a breakthrough by Anderson~\etal~\cite{Anderson2002},
who gave an $O(n^4)$-time algorithm for the \emph{successful-queries} variant,
in which only searches for keys stored in the tree are supported.  
Their work was motivated by an application to efficient message dispatching
in object-oriented programming languages, addressed earlier
by Chambers and Chen~\cite{Chambers:1999:EMP:320384.320407} (see also~\cite{nagaraj_optimal_1997})
who provided an $O(n^2)$-time heuristic (non-optimal) algorithm for constructing
two-way-comparison trees. Independently, Andersson~\cite{andersson1991note} had earlier looked at speeding up search trees by replacing three-way comparisons with two-way ones.


\paragraph{Difficulties introduced by two-way comparisons.}  
Given a query $\query$,
the search for $\query$ in a search tree starts at the root and traces a path to a leaf,
at each node comparing $\query$ to the node's key,
then following the edge corresponding to the outcome.
A \emph{three-way} comparison has three possible outcomes:
the query is less than, equal to, or greater than the node's key. 
In any three-way-comparison tree,
the set of queries reaching any internal node $N$ is an open interval between some two keys.
Each subtree solves a subproblem defined by some such inter-key interval.
This leads to a simple dynamic program with $O(n^2)$ subproblems
and to a simple $O(n^3)$-time algorithm,
the derivation of which is a standard exercise~\cite
[\S15.5]{Cormen:2009:IAT:1614191},~\cite[Ex.\,6.20]{dasgupta_algorithms_2006}.
Knuth's celebrated result improves the running time to $O(n^2)$.

In contrast (as noted e.g.~in~\cite{Anderson2002}),
algorithms for \emph{two-way}-comparison search trees face the following obstacle.
After a few equality tests, the subproblem that remains
is defined by some inter-key interval with \emph{holes} --- each hole corresponding to an earlier equality test.
(For example, in tree (b) of Figure~\ref{fig: three-way vs two-way},
the set of queries reaching node $\compnode < 5$ is $[3,\infty) - \{7\}$.)
The resulting dynamic-programming formulation can have $\Theta(2^{n})$ subproblems, as
the query sets that can arise can contain any subset of the keys.


\begin{figure}[t]
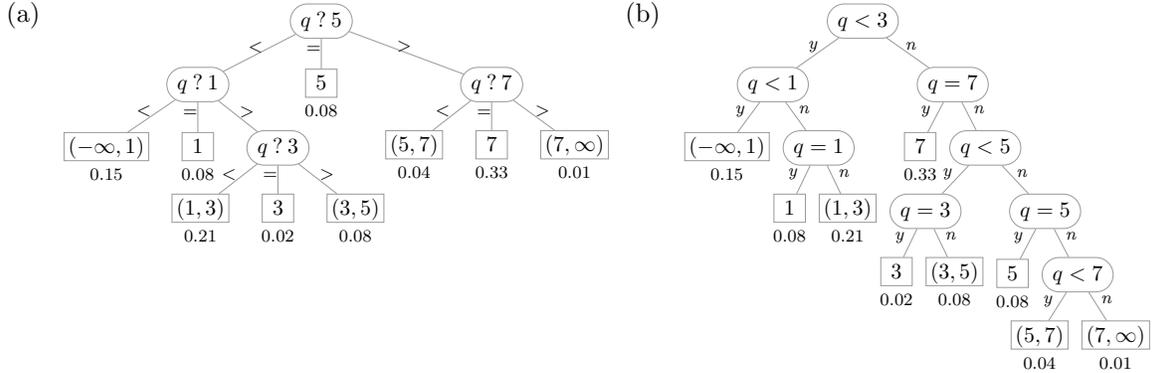

 \centering
  \newcommand{\CMP}{\mathbin{?}}
  \centering
  \small
  \noindent
  (a)
  \adjustbox{valign=t,scale = 0.82}{
    \begin{ThreeWayTree}
      [{$q \CMP 5$}, s sep=3.5em,
        [{$q \CMP 1$}, 
          [{$(-\infty,1)$}, weight={0.15}]
          [1, calign with current, weight={0.08}]
          [{$q \CMP 3$}
            [{$(1, 3)$}, weight={0.21}]
            [3, calign with current, weight={0.02}]
            [{$(3,5)$}, weight={0.08}]
            ]
          ]
        [5, calign with current, weight={0.08}]
        [{$q \CMP 7$}
          [{$(5,7)$}, weight={0.04}]
          [7, calign with current, weight={0.33}]
          [{$(7,\infty)$}, weight={0.01}]
          ]
        ]
    \end{ThreeWayTree}
  }
  (b)
  \adjustbox{valign=t,scale = 0.82}{
    \begin{TREE}
      [{$q<3$}
        [{$q<1$}
          [{$(-\infty,1)$}, weight={0.15}]
          [{$q=1$}, 
            [1, weight={0.08}]
            [{$(1,3)$}, weight={0.21}]
            ]
          ]
        [{$q=7$}
          [7, weight={0.33}]
          [{$q<5$}
            [{$q=3$}, 
              [3, weight={0.02}]
              [{$(3,5)$}, weight={0.08}]
              ]
            [{$q=5$}, 
              [5, weight={0.08}]
              [{$q<7$}, 
                [{$(5,7)$}, weight={0.04}]
                [{$(7,\infty)$}, weight={0.01}]
                ]
              ]
            ]
          ]
        ]
    \end{TREE}
  } 
  \caption{
  	(a) A \emph{three-way-comparison} search tree and
 	(b) a \emph{two-way-comparison} tree.
    Both trees have keys $\{1,3,5,7\}$. The number below each leaf represents its frequency.
	The cost of tree~(a) is 2.23, the cost of tree~(b) is $2.8$.
  }\label{fig: three-way vs two-way}
\end{figure}


\paragraph{The most-likely-key property.} 
To achieve their polynomial-time algorithm for two-way comparison trees, 
Anderson \etal circumvent this obstacle in two steps.
First, using an elegant \emph{side-weights} argument,
they show that every instance has an optimal tree $T$ with what we call the \emph{most-likely-key property}:
\begin{description}
\item(MLK) \emph{In each equality-test node $N$ in $T$,
    the test is to some key of maximum weight
    among keys reaching $N$.
  }
\end{description}
Their proof is for the successful-queries variant
but extends directly to the unrestricted variant.
The result implies that it suffices to consider only trees with the MLK property.

If all key weights are distinct, then, at any node $N$ in such a tree $T$, the set of
queries that reach $N$ is an inter-key interval
with the $g$ heaviest keys removed, for some $g \in\{0,1,\ldots,n\}$.
There are $O(n^3)$ such query sets (intervals with such keys removed)
so this yields a dynamic program with $O(n^3)$ subproblems,
and a simple $O(n^4)$-time algorithm for instances with distinct weights.

When the weights are not distinct, however, the MLK property gives no guidance about how to break ties.
Even in trees with the MLK property, exponentially many distinct query sets can arise,
so the dynamic program still has exponential size.
To work around this counter-intuitive obstacle, Anderson \etal
then resort to a lengthy argument that establishes
``thresholds on the frequency of an item that can occur as an equality comparison at the root of an optimal tree.''
This significantly complicates their analysis and their algorithm.


\paragraph{New results.}
Our first result is a structural theorem (Theorem~\ref{thm: root RMLK}), which
implies that
there is an optimal tree with what we call the \emph{refined MLK property}:
\begin{description}
\item(RMLK)
  \emph{In each equality-test node $N$ in $T$,
    the test is to the key $k_b$ of maximum weight among keys
    reaching $N$, breaking ties by maximizing the index $b$.}
\end{description}
Breaking ties by the index $b$ (above) is for convenience, in fact any tie-breaking rule will work. 

The precise statement of Theorem~\ref{thm: root RMLK} is somewhat delicate.
For example, it is \emph{not} the case that any tree with an equality test at the root
can be converted into an equally good tree with an equality test to any given maximum-weight key.
(See the remarks following the proof of the theorem.)
But, with the correct theorem statement in hand, the proof is straightforward,
by augmenting the side-weights argument of Anderson \etal
with a perturbation argument.

The theorem implies that, whether or not the weights are distinct,
it suffices to consider trees with the RMLK property.
Restricting to such trees yields a simple dynamic program with $O(n^3)$ subproblems,
each solvable in $O(n)$ time,
yielding an $O(n^4)$-time algorithm (Corollary~\ref{corollary: main}).
The algorithm is essentially the algorithm for distinct-weights instances, modified to break ties arbitrarily.
Our formulation is simple and easy to implement ---
Appendix~\ref{sec: code} gives a twenty-two-line implementation in Python.
This is the first polynomial-time algorithm to be proven correct for the unrestricted variant.


\begin{figure}[t]
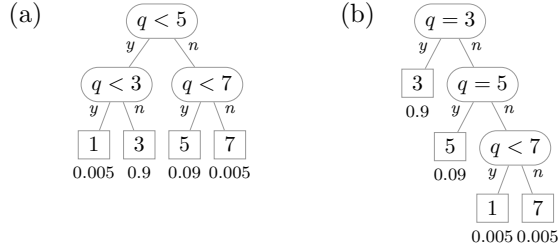

 \centering
  \newcommand{\CMP}{\mathbin{?}}
  \centering
  \small
  \noindent
  (a)
\adjustbox{valign=t,scale = 0.82}{
    \begin{TREE}
      [{$q<5$}
        [{$q<3$}
          [1, weight={0.005}]
          [3, weight={0.9}]
         ]
        [{$q<7$}
          [5, weight={0.09}]
          [7, weight={0.005}]
           ]
        ]
    \end{TREE}
  } 
\hspace*{.3in}
  (b)
  \adjustbox{valign=t,scale = 0.82}{
    \begin{TREE}
      [{$q=3$}
        [3, weight={0.9}]
        [{$q=5$}
          [5, weight={0.09}]
          [{$q<7$}
               [1, weight={0.005}]
               [7, weight={0.005}]
            ]
          ]
        ]
    \end{TREE}
  } 
  \caption{
Optimal {two-way-comparison} trees for the \emph{successful-queries} variant
  	(a) using only inequalities, and
 	(b) when `$=$'s are also permitted.
    Both trees have key set  $\{1,3,5,7\}$. The number below each leaf represents its frequency.
	The cost of tree~(a) is 2, the cost of tree~(b) is $1.11$.
}\label{fig: two-way successful queries}
\end{figure}



\paragraph{Other related work.}
In 1994, Spuler~\cite{Spuler1994Paper,Spuler1994Thesis} suggested an algorithm for computing
optimal two-way comparison trees, based on a result by Huang~\etal~\cite{StephenHuang1984} 
for a different type of search tree. But Huang \etals related result,
and the recurrences underlying Spuler's proposed algorithm,
are demonstrably wrong~\cite{chrobak_etal_huangs_algorithm_2018}.
Independently, Spuler also conjectured that optimal trees have the MLK property,
and based on that conjecture proposed an additional algorithm, without proof of correctness,
running in time $O(n^5)$.
As discussed earlier, in 2002 Anderson \etal proved the conjecture for the successful-queries variant
and obtained an $O(n^4)$-time algorithm for it.
  To our knowledge, this is the only previously published proof of correctness
  for any polynomial-time  algorithm for optimal two-way-comparison search trees,
for any variant with equality tests and inequality comparisons.

\smallskip

The technical obstacle introduced by non-distinct weights
also arises in finding optimal so-called \emph{binary split trees}:
for distinct-weight instances, an optimal binary split tree can be found in $O(n^4)$ time,
but for arbitrary instances the best bound known is $O(n^5)$~\cite{Hester1986}.
For \emph{generalized binary split trees} (see~\cite{StephenHuang1984,chrobak_etal_huangs_algorithm_2018})
no correct polynomial-time algorithm is yet known.

\smallskip

There are various other models for comparison-based search trees in the literature.
As one example, one can consider only inequality comparisons. In trees with
three-way comparisons that would correspond to giving every key weight zero
(so only non-key queries have positive probability).
As noted by Knuth,  this model is 
equivalent to alphabetic coding, solved by Gilbert and Moore in $O(n^3)$ time~\cite{Gilbert1959},  
later improved to $O(n \log n)$ time by Hu and Tucker~\cite{Hu1971} and,
independently, by Garsia and Wachs~\cite{Garsia1977}. 
As noted by Anderson \etal this is also equivalent to 
finding the optimal two-way-comparison tree
in the successful-queries case that only uses ``$<$'' comparisons.                              
For more background information on different models and algorithms (including approximations)
for search trees we refer the reader to the 1997 tutorial by Nagaraj~\cite{nagaraj_optimal_1997}.
     



\paragraph{Formal problem statement.} 
Let $k_1, k_2, \ldots, k_n$ denote the $n$ given keys, in increasing order.
Let $\beta_b\ge 0$, for $b = 1,2,\ldots,n$,
be the weight of key $k_b$ and $\alpha_a\ge 0$, for $a = 0,1,\ldots,n$, be
the weight of interval $(k_a, k_{a+1})$, where for convenience $k_0$ and $k_{n+1}$ (which are not keys)
represent $-\infty$ and $\infty$.
In standard applications the weights represent query frequencies or probabilities.
Here each weight can be an arbitrary non-negative number.

As illustrated in Figure~\ref{fig: three-way vs two-way}(b),
a \emph{two-way-comparison search tree} for the instance is a rooted binary tree $T$,
where each non-leaf node $N$ represents a comparison ($=$ or $<$)\footnote
{For simplicity of presentation, we don't allow the $>$ (or equivalently $\le$) comparison.
  See the discussion below.}
to some key $k_b$. $T$ has $2n+1$ leaves, $n$ of which are
associated one-to-one with the keys $k_b$ and the remaining $n+1$ with the intervals $(k_a, k_{a+1})$.
Given any value $\query$,  the search for $\query$ in $T$ must correctly
\emph{identify} $\query$, in the following sense. The search starts at the root and 
proceeds down the tree, following the edges corresponding to the outcomes of the comparisons,
and it terminates either at the leaf associated with a key $k_b$, if $\query = k_b$,
or at the leaf associated with the interval $(k_a, k_{a+1})$ containing $\query$, if $\query$ is not a key.

For each node $N$ in $T$, let $T_N$ denote the subtree rooted at $N$.
The \emph{weight of subtree $T_N$}, denoted $\nodeweight{N}$, is the total weight of the
  keys $k_b$ and intervals $(k_a, k_{a+1})$ in the set of queries reaching $N$
  (these are the keys and intervals associated with the leaves of $T_N$).

The \emph{cost} of $T$ is the weighted depth of its leaves, that is:
\begin{equation}
  \cost (T) \;=\; \sum_{ \ell\in\leaves{T} } \depth_T(\ell) \times\, \nodeweight{\ell}
  			\;=\; \sum_{ N\in T-\leaves{T} } \nodeweight{N},
          \label{eq: cost} 
\end{equation}
where $\leaves{T}$ is the set of leaves in $T$ and
$\depth_T(\ell)$ is the number of edges on the path from the root to $\ell$ in $T$. 

A search tree $T$ is called \emph{optimal} for a given instance if it has
minimum \nealrevisedtext{cost} among all search trees for the given instance.
The objective is to compute such an optimal tree. 


  \paragraph{The successful-queries variant.}
  The problem definition above is for the ``full'' variant, with unrestricted queries.
  In the \emph{successful-queries} variant, illustrated in~Figure~\ref{fig: two-way successful queries},
	only queries to the given keys $k_1,k_2,\ldots,k_n$ are allowed.
  An instance is specified by just those keys and their frequencies $\beta_1,\beta_2,\ldots,\beta_n$,
  and the problem is to find the minimum cost of any tree that handles $\Queries=\{k_1,\ldots,k_n\}$.
  (Section~\ref{sec: dynamic programming algorithm} formally defines ``handles''.)
  While our algorithm's description in Section~\ref{sec: dynamic programming algorithm}
	assumes the full variant,
  it can be adapted directly to the successful-queries variant, retaining its running time of $O(n^4)$,
  by a straighforward modification of the dynamic program\footnote{%
    {%
  Note that the successful-queries variant is \emph{not} a special case of the full variant, since, by definition,
  any search tree for this variant must contain leaves representing all inter-key intervals $(k_a,k_{a+1})$, even 
  though their probabilities $\alpha_a$ are $0$.}
  }.
  This adapted algorithm has the same asymptotic running time as Anderson et~al.'s
 	(who considered only the successful case), but  is much simpler.


\paragraph{Other comparison operators.}
{
  Similarly, for the sake of exposition, and following the model in~\cite{Anderson2002},
  we allow only the two comparison operators $\{<,=\}$.
  As explained later (see the end of Section~\ref{sec: dynamic programming algorithm}),
 our algorithm and its analysis extend naturally to the model with comparison set $\{<, \le, =\}$.
 }
 Due to symmetries and equivalences, this is the only other set of operators 
of interest for the full variants of the problem (allowing unsuccessful queries) that include equality comparisons. 
{%
For example, operator $\ge$ can be replaced by operator $<$ by swapping the subtrees of each node with the $\ge$
comparison.
}
We note, in passing, that Hu and P.~A.\ Tucker~\cite{hu_optimal_1998} have shown that the full variant
with only comparison operators in $\{<, \le\}$ allowed can be solved in time
$O(n \log n)$ by reducing it to an instance of alphabetic coding with $2n+1$ keys and applying the algorithm of
Hu and A.~C.\ Tucker~\cite{Hu1971}.




\section{The Dynamic-Programming Algorithm}\label{sec: dynamic programming algorithm}


Within any tree having the RMLK property, at any node $N$,
the set $\Queries$ of queries reaching $N$ has the following form:
$\Queries$ consists of the queries lying in some inter-key interval,
minus some number of heaviest keys removed.
The algorithm will solve one subproblem for each set $\Queries$ of this form.
Next we explicitly define these sets and their subproblems.

Fix a permutation
$\key {\pi(1)}, \key {\pi(2)},\ldots{},\key {\pi(n)}$ of the keys
that orders them by nondecreasing weight, breaking ties arbitrarily.
For each key $\key {\pi(r)}$, call $r$ the \emph{rank} of the key.
A set $\Queries$ is \emph{valid} if either (i) $\Queries$ is a singleton key, of the form $\Queries = \braced{k_b}$,
or (ii) $\Queries$ is an inter-key interval with some largest-rank keys removed
--- that is, for some $i, j \in \{0,1,\ldots, n+1\}$
with $i<j$,
and $h\in\{0, 1,\ldots,n\}$, 
the set $\Queries$ equals $\querysubset i j h$, defined as
\lineNumberPatch
\begin{equation*}
  \querysubset i j h\; = \;  [\key i, \key j) - \braced{ \key {\pi(h+1)} , \key {\pi(h+2)} , \ldots , \key {\pi(n)}, \key 0}.
\end{equation*}
(These sets are not necessarily distinct
--- for example, it is possible that
$\querysubset{i}{j}{h} = \querysubset{i}{j}{h'}$
even though, say, $h\ne h'$. 
Removing $k_0$ is needed only to deal with the case when $i=0$.)

Note that $\querysubset i j h$ is a union of some keys
and ``failure'' intervals $(k_a, k_{a+1})$. Let $\treeweight{\querysubset i j h}$ denote the
total weight of these keys and intervals. Given a search tree $T$, if those keys and intervals
are also the ones associated with $T$'s leaves, say that \emph{$T$ handles $\querysubset i j h$}
(so, for each $\query\in \querysubset i j h$, the tree correctly identifies $\query$).
Define $\optcost(i,j,h)$ to be the minimum cost of any tree that handles $\querysubset i j h$.
Note that $\querysubset 0 {n+1} n = (-\infty, +\infty)$, so the goal is to return $\optcost(0, n+1, n)$.
(If desired, an optimal tree can also be constructed in the standard fashion.)

The algorithm computes $\optcost(0, n+1, n)$
using the following recurrence relation: 


\begin{enumerate}[(a)]

\item
  If $j=i+1$ and $k_i\not\in \querysubset i j h$,
  then $\querysubset i j h = (k_i, k_{i+1})$
  so $\optcost(i,j,h) = 0$. 

\item
  If $h> 0$ and $\key{\pi(h)} \notin \querysubset i j h$,
  then $\querysubset i j h =\querysubset {i} {j} {h-1}$,
  so $\optcost(i,j,h) = \optcost(i,j,h-1)$.

\item Else,
  \(\displaystyle
  \optcost(i,j,h) \,=\,
  \treeweight {\querysubset i j h} + 
    \min\begin{cases}
      \optcost(i,j,h-1) \textit{~~~(if $h>0$, else $+\infty$),}
      \\[0.3em]
      \min \big\{\, 
      \optcost(i,b,h) + \optcost(b,j,h) : i < b < j\,\big\}. 
    \end{cases}
  \)
\end{enumerate}


\paragraph{Correctness.}
For any tree $T$ that handles $\querysubset i j h$,
note that $\cost(T)$ is zero if $T$ is a leaf,
and otherwise $\cost(T) = \treeweight{\querysubset i j h}+ \cost(T_L) + \cost(T_R)$,
where $T_L$ and $T_R$ are the left and right subtrees.

By inspection, the cases in the recurrence relation are exhaustive.
Likewise, Cases (a) and (b) are correct. For Case (c),
the two terms in the right-hand side 
correspond to creating a tree whose root does either
(i) an equality-test with key $\key {\pi(h)}$
(if $h>0$ so that $\querysubset i j h$ contains some keys,
with $\key {\pi(h)}$ being the one of largest rank),
or (ii)  an inequality comparison with some key $\key b$.
By Theorem~\ref{thm: root RMLK} in Section~\ref{sec: structure theorem}
there is an optimal tree of this form that handles $\querysubset i j h$,
so the recurrence is correct in Case (c).
(Note that if $j=i+1$, the range of the minimum in the second term is empty,
and there no trees of type (ii) above.  In this case,
we take the minimum to be infinite, so the recurrence remains correct.)


\paragraph{Running time.}  
The number of subproblems $\querysubset i j h$ is $O(n^3)$.
For each subproblem, by inspection of the recurrence,
$\optcost(i, j, h)$ can be computed in $O(n)$ time.
Thus the running time is $O(n^4)$. This proves Corollary~\ref{corollary: main}:

\begin{corollary}\label{corollary: main}
  There is an $O(n^4)$-time algorithm
  for computing optimal two-way-comparison search trees (with unrestricted queries).
\end{corollary}


{
\paragraph{Extension to operators $\{=,<, \le \}$.}
The algorithm and its proof of correctness extend naturally to
the model with three comparison operators, $\{=,<,\le \}$,  potentially allowing lower 
tree costs (even for the case $n = 1$). The overall principle of the algorithm remains the same,
although the presence of two inequality operators introduces minor technical complications.
The modified algorithm will have four types of subproblems that correspond to four types of
inter-key intervals: $(k_i,k_j)$, $[k_i,k_j)$, $(k_i,k_j]$, and $[k_i,k_j]$, each
with some number $h$ of holes. In the recurrence for each type of subproblem,
the minimum in (c) will have three choices that correspond to operators $=$, $<$ and $\le$.
As for the correctness proof, the statement of Lemma~\ref{lemma: root MLK} needs to be modified to
say in part~(b) that ``The root of $T$ does an inequality comparison''.
The proofs of Lemma~\ref{lemma: root MLK}, Lemma~\ref{lemma: side weights},
and Theorem~\ref{thm: root RMLK} apply as presented.
}





\section{Statement and Proof of Theorem~\ref{thm: root RMLK}}\label{sec: structure theorem}


Fix the keys $k_1, k_2, \ldots, k_n$ and their associated weights. Let $\Queries$ be any valid set,
and consider the subproblem of finding a minimum-cost (two-way-comparison) tree that handles $\Queries$.

\begin{theorem}\label{thm: root RMLK}
  Fix any key $k_b$ of largest weight among keys in $\Queries$.
  Then some optimal tree $T$ satisfies one of the following three conditions:
  \begin{enumerate}[(i)]
  \item $T$ consists of a single leaf.
  \item The root of $T$ does a less-than comparison.
  \item The root of $T$ does an equality test to $k_b$.
  \end{enumerate}
\end{theorem} 
(See also the stronger Theorem~\ref{thm: root stronger} in Appendix~\ref{sec: root stronger}, with a different proof.)

\newcommand{\lemmaRootMLK}{          
  Some optimal tree $T$ satisfies one of the following three conditions:
  \begin{enumerate}[(a)]
  \item $T$ consists of a single leaf. 
  \item The root of $T$ does a less-than comparison.
  \item The root of $T$ does an equality test to some key $k_b$ of maximum weight among keys in $\Queries$.
  \end{enumerate}
}
\begin{proof}
  The rest of this section gives the proof.
  It uses the following lemma as a black box.
  \begin{lemma}\label{lemma: root MLK}
    \lemmaRootMLK
  \end{lemma}

  (Anderson \etal prove essentially the same lemma (their Corollary 3),
  but for the successful-queries variant. Their proof uses their side-weights technique
  and extends essentially unchanged to our setting.
  For completeness we include the proof in Appendix~\ref{sec: MLK proof}.)

  \smallskip
  
  {
    The idea of the proof is to
    infinitesimally perturb the weights
    to make $k_b$ the uniquely largest-weight key in $\Queries$.
    Lemma~\ref{lemma: root MLK} implies that there is a tree $T$
    that is optimal for the perturbed instance
    and has one of the Properties (a)--(c),
    which (given the perturbation) implies that $T$ has one of the Properties (i)--(iii).
    The perturbation is sufficiently small so that any optimal tree ($T$, in particular) for the perturbed instance
    must also be optimal for the original instance.
    (Similar arguments have been used for other problems, such as the Minimum Spanning Tree problem,
    to extend algorithms for distinct-weights instances
    to algorithms for general instances,
    see e.g.~\cite[\S 4.5]{kleiberg_tardos_algorithm_design_2006}.)
    Here are the details.

  }
  For each integer $i\ge 1$, define a new instance with the same set of keys but 
  with  the weight of key $k_b$ increased by $1/i$. That is, let 
  $\beta^i_b = \beta_b+1/i$ and $\beta_{a}^i =\beta_{a}$ for all $a\ne b$.
  By Lemma~\ref{lemma: root MLK}, 
  there is an optimal tree $T^i$ for the modified instance 
    such that $T^i$ satisfies one of (a)--(c) of Lemma~\ref{lemma: root MLK} for the modified instance,
    and therefore satisfies one of (i)--(iii) of the theorem for the original instance.
  Let tree $T$ be such that $T=T^i$ for infinitely many $i$. 
  (Such a $T$ must exist, because the possible trees for each modified instance
  are the same, and there are finitely many of them.)
  Then $T$ is a possible tree for the original instance and has one of Properties (i)--(iii), as desired.
  To finish we show that $T$ is optimal for the original instance.

  Let $T^*$ be an optimal tree for the original instance. 
  For all $i$ such that $T=T^i$, 
  \lineNumberPatch 
  \begin{align*}
    \cost(T, \beta)
    &{} = \cost(T, \beta^i) - \depth_{T}(k_b)/i && \text{by definition of $\beta^i_a$'s and Equation~\eqref{eq: cost}} \\
    &{} \le \cost(T, \beta^i) && \\
    &{} = \cost(T^i, \beta^i) && \text{because $T = T^i$ for this $i$} \\
    &{} \le \cost(T^*, \beta^i) && \text{because $T^i$ is optimal for $\beta^i_a$'s} \\
    &{} = \cost(T^*, \beta) + \depth_{T^*}(k_b)/i && \text{by definition of $\beta^i_a$'s and Equation~\eqref{eq: cost}.}
  \end{align*}  
  (In this derivation, arguments $\beta$ or $\beta^i$ in the cost function
  indicate which key weights are used to compute it.)
  This holds as $i\rightarrow\infty$, so $\cost(T, \beta) \le \cost(T^*, \beta)$,
  and thus $T$ is also optimal for the original instance.
\end{proof}


\paragraph{Remarks.}
The correct statement of Theorem~\ref{thm: root RMLK} is somewhat delicate.
Consider the following statement:
\emph{If there is an optimal tree whose root does an equality test,
  then, for any maximum-weight key $k$ in $\Queries$,
  there is an optimal tree whose root does an equality test to $k$.}
This seemingly similar statement is false.
Figure~\ref{fig: discontinuity} gives a counterexample. 

In that figure, keys 1 and 2 both have maximum weight
among keys reaching the right child of the root.
Trees (a) and (b) are optimal.
In (a), the right child of the root does an equality test to key $1$,
but no optimal subtree for that subproblem
(the subproblem with valid set $\Queries=[1,\infty)$)
has an equality test to key $2$ at the root
--- tree (c) has minimum cost among trees with such a subtree.
This counter-example also holds in the problem variant
where all comparisons in $\{<, \le, =\}$ are allowed.\footnote
{For the succesful-queries variant no such counter-example is possible.
  By Theorem~\ref{thm: root stronger} in Appendix~\ref{sec: root stronger},
  if some optimal tree has an equality test at the root, then for \emph{any} maximum-weight key 
  there is an optimal tree that uses this key in the equality-test at the root.
  This stronger property, however, does not seem to have any algorithmic implications.}

\begin{figure}
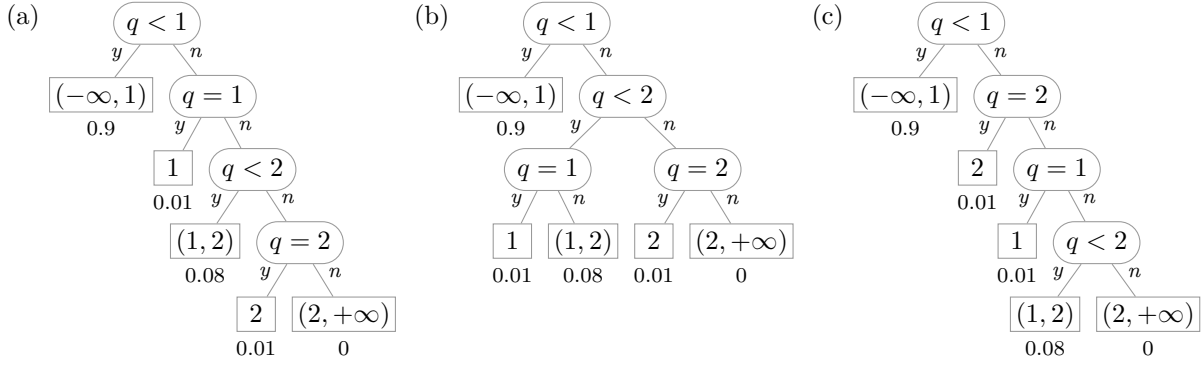
\small

  \renewcommand{\treeLabelSize}{\scriptsize}
  
  (a)
  \hspace*{-1em}
  \adjustbox{valign=t}{
    \begin{TREE}
      for tree = { l sep=0em }, 
      [{$q<1$}
        [{$(-\infty, 1)$}, weight={0.9}]
        [{$q=1$}
          [1, weight={0.01}]
          [{$q<2$}
            [{$(1, 2)$}, weight={0.08}]
            [{$q=2$}
              [2, weight={0.01}]
              [{$(2,+\infty)$}, weight={0}]
            ]
          ]
        ]
      ]
    \end{TREE}
  }
  (b)
  \hspace*{-1em}
  \adjustbox{valign=t}{
    \begin{TREE}
      for tree = { l sep=0em }, 
      [{$q<1$}
        [{$(-\infty, 1)$}, weight={0.9}]
        [{$q<2$}
          [{$q=1$}
            [1, weight={0.01}]
            [{$(1, 2)$}, weight={0.08}]
            ]
          [{$q=2$}
            [2, weight={0.01}]
            [{$(2,+\infty)$}, weight={0}]
            ]
          ]
        ]
    \end{TREE}
  }
  (c)
  \hspace*{-1em}
  \adjustbox{valign=t}{
    \begin{TREE}
      for tree = { l sep=0em }, 
      [{$q<1$}
        [{$(-\infty, 1)$}, weight={0.9}]
        [{$q=2$}
          [2, weight={0.01}]
          [{$q=1$}
            [1, weight={0.01}]
            [{$q<2$}
              [{$(1, 2)$}, weight={0.08}]
              [{$(2,+\infty)$}, weight={0}]
              ]
            ]
          ]
        ]
    \end{TREE}
  }
  
  \caption{
    Trees for a full instance with $\Keys=\{1,2\}$,
    $(\beta_1,\beta_2)=(0.01, 0.01)$ and
    $(\alpha_0,\alpha_1,\alpha_2) =(0.9, 0.08, 0)$.
    As discussed in the remarks,
    the subproblem with valid set $\Queries=[1,\infty)$
    has an optimal solution with an equality test to key $1$,
    but none with an equality test to key $2$, even though $\beta_1 = \beta_2$.
  }\label{fig: discontinuity}
\end{figure}

\smallskip 

The proof of Theorem~\ref{thm: root RMLK}  
perturbs the weights infinitesimally, giving $k_b$ the unique largest weight
and (via Lemma~\ref{lemma: root MLK}) forcing any root equality test to be to key $k_b$.
It then observes that, because the minimum tree cost is a continuous function of the weights,
any tree that is optimal for the (infinitesimally!) perturbed instance
must also be optimal for the original instance.   
This infinitesimal perturbation can make
\emph{all} previously optimal trees with equality tests at the root non-optimal.
(For example, in Figure~\ref{fig: discontinuity}, after an infinitesimal increase in $\beta_2$, 
only tree (b) would be optimal.)

\smallskip

The RMLK property from the introduction is a corollary of the theorem.
To prove it, let $T$ be the optimal tree for $\Queries$, with root as specified in the theorem.
Then (using induction on $n$) replace each of the two subtrees
by an optimal one (for its subproblem) having the RMLK property.
The resulting tree $T'$ is also an optimal tree for $\Queries$, and has the RMLK property.

\smallskip

In fact, there must be an optimal tree with the following stronger property.
\emph{At any node $N$ that does an equality test,
  if there are $d$ largest-weight keys in the search space reaching $N$,
  then the top $d$ nodes on the right spine in $T_N$ do equality tests to those $d$ keys in some order
  (and permuting that order arbitrarily preserves optimality).}
See Corollary~\ref{cor: root stronger} in Appendix~\ref{sec: root stronger}.



\smallskip

{
\paragraph{Acknowledgements.}
We are very grateful to anonymous reviewers of the journal version for insightful comments
that helped us improve the presentation.
}



\bibliographystyle{plainurl} 
\bibliography{optimal_search_trees}



\section*{Appendices}

\appendix 

\section{Proof of Lemma~\ref{lemma: root MLK}}\label{sec: MLK proof}

This appendix proves Lemma~\ref{lemma: root MLK}.
First we state Anderson \etals side-weights lemma~\cite[Lemma 2]{Anderson2002},
and prove it in our setting.
(The proof is identical.)
Fix the keys $k_1, k_2, \ldots, k_n$ and associated weights.
Let $\Queries$ be any valid set
and consider the problem of finding
a minimum-cost tree among trees that handle $\Queries$.
Let $T$ be any optimal tree for this problem.

Define the \emph{side-weight} of any node $N$ in $T$, denoted $\sw N$, as follows.
If $N$ is a leaf, then $\sw N = 0$.
If $N$ does an equality test to a key, then $\sw N$ is the weight of that key.
Otherwise ($N$ is an inequality-comparison node) 
$\sw N = \min\braced{\weight{L},\weight{R}}$,
where $L$ and $R$ are $N$'s left and right children.
The utility lemma states that side-weights are monotone along paths from the root:


\begin{lemma}[{\cite[Lemma 2]{Anderson2002}}]\label{lemma: side weights}
  If $P$ is the parent of $N$ in $T$, then $\sw P\ge \sw N$.
\end{lemma}       


\begin{proof}
  If $N$ is a leaf, it has side-weight 0 and the property holds, so assume that $N$ is not a leaf.
  There are four cases according to whether at each of $P$ or $N$ the comparison done
  is an equality or an inequality.
  Recall that $T$ is an optimal tree.  
         
  \mycase{1} \emph{Both $P$ and $N$ are inequality comparisons.}
  Assume without loss of generality that $N$ is the right child of $P$.
  Let $T_1$ be the subtree rooted at the child of $P$ that is not $N$.
  Let $T_2$ and $T_3$ be the subtrees rooted, respectively, at the left and right children of $N$.
  Let $\mu_i$ denote the weight of $T_i$ (for $i\in\{1,2,3\})$.
  Then $\sw P = \min\{\mu_1,\mu_2+\mu_3\}$
  and $\sw N = \min \{\mu_2,\mu_3\}$.
  Now do a left rotation at $P$,
  so that $N$ takes the place of $P$,
  $P$ becomes $N$'s left child, and $T_2$ becomes the right subtree of $P$.
  While $T_2$ stays at the same depth,
  $T_1$ moves down and $T_3$ moves up, each by one level. 
  The increase in cost is $\mu_1 - \mu_3$.
  This must be non-negative, as $T$ is optimal, so $\mu_1 \ge \mu_3$.
  This implies $\min\{\mu_1,\mu_2+\mu_3\}\ge \min\{\mu_2,\mu_3\}$,
  that is, $\sw P \ge \sw N$.

  \mycase{2} 
  \emph{Node $P$ is an inequality comparison and $N$ is an equality comparison.}
  Let $T_1$, of weight, say, $\mu_1$, be the subtree rooted at the child of $P$ which is not $N$.
  Let $\mu_2$ be the weight of the left subtree, say $T_2$, of $N$ (consisting of a leaf for the equality-test key of $N$).
  Let $\mu_3$ be the weight of the right subtree, say $T_3$ of $N$.
  We have $\sw P  = \min\{\mu_1, \mu_2 + \mu_3\}$ and $\sw N = \mu_2$.
  Modify $T$ by moving the equality test at $N$ just above $P$.
  Tree $T_1$ moves down one level, $T_3$ stays at the same level, yet $T_2$, of weight $\mu_2$, moves up one level.
  The net increase in cost is $\mu_1 - \mu_2$.
  Since $T$ is optimal, this is non-negative, so $\mu_1 \ge \mu_2$.
  Thus $\min\{\mu_1, \mu_2 + \mu_3\} \ge \mu_2$.
  That is, $\sw P  \ge \sw N$.

  \mycase{3}
  \emph{Both $P$ and $N$ are equality comparisons.}
  Swap the comparisons in $P$ and $N$.
  The increase in cost is $\sw P - \sw N$.
  Since $T$ is optimal, this is non-negative, so $\sw P \ge \sw N$. 

  \mycase{4} 
  \emph{Node $P$ is an equality comparison and $N$ is an inequality comparison.}
  Let $\mu_1$ be the weight of $P$'s left-subtree $T_1$ (consisting of a leaf for the equality-test key of $P$).
  Let $T_2$ and $T_3$ of weight $\mu_2$ and $\mu_3$, respectively, be the subtrees hanging off $N$.
  Move the equality-test node $P$ down just above the appropriate child of $N$.
  Then exactly one of $T_2$ and $T_3$ moves up while $T_1$ moves down. 
  The increase in cost is either $\mu_1 - \mu_2$ or $\mu_1 - \mu_3$.
  Since $T$ is optimal, the increase is non-negative, so $\mu_1 \ge \min\{\mu_2,\mu_3\}$. 
  That is, $\sw P  \ge \sw N$.
\end{proof}

Next is the proof of Lemma~\ref{lemma: root MLK}.
The proof is essentially the same as Anderson \etals proof of~\cite[Corollary 3]{Anderson2002},
which is for successful-queries instances but extends directly to our setting.


\setcounter{lemma}{0}

\begin{lemma}
  \lemmaRootMLK
\end{lemma}
\begin{proof}
  Let $T$ be an optimal tree with root $N$.
  In the case that $N$ does an inequality comparison, we are done.
  So assume $N$ does an equality test to some key $k_b$, necessarily in $\Queries$.
  Let $k_a$ be any largest-weight key in $\Queries$.
  The parent of the leaf for $k_a$ is some node $P$ in $T$.
  First, if $P$ doesn't do an equality test to $k_a$,
  replace the comparison at $P$ by an equality test to $k_a$,
  without otherwise changing $T$.
  This preserves correctness and optimality.
  Now, in the case that $P$ is the root $N$, we are done.
  In the remaining case, by Lemma~\ref{lemma: side weights}, 
  along the path from $N$ to $P$, the side-weights are non-increasing,
  so $\beta_b = \sw N \ge \sw P = \beta_{a}$,
  and $k_b$ is a largest-weight key in $\Queries$.
\end{proof}

\section{A stronger alternative to Theorem~\ref{thm: root RMLK}}\label{sec: root stronger}

\emph{
  The results in this section
  do not appear in the ACM TALG journal version of this paper~\cite{chrobak_etal_simple_bcst_algorithm_2021}.
}

Fix an instance of either the full variant or the successful-queries variant
(as described in the introduction and illustrated in Figure~\ref{fig: two-way successful queries}(b)).
Consider the subproblem of finding a minimum-cost (two-way-comparison) tree that handles
any given valid set $\Queries$.
(The terms \emph{valid} and \emph{handles}
are defined in Section~\ref{sec: dynamic programming algorithm}.)

\begin{theorem}\label{thm: root stronger}
  For either problem variant, at least one of the following conditions holds:
  \begin{enumerate}[(i)]
  \item The optimal tree consists of a single leaf.
  \item There is an optimal tree whose root does an inequality comparison.
  \item For any ordering of the max-weight keys in $\Queries$,
    there is an optimal tree that starts by doing equality tests to each of those keys in that order
    (before doing any inequality tests or equality tests to keys of lesser weight).
  \end{enumerate}
  For the successful-queries variant, 
  if any optimal tree does an equality test at the root,
  then either Condition (iii) above holds,
  or all keys in $\Queries$ have the same weight
  and, for any ordering of those keys,
  there is an optimal tree that consists solely of equality tests (in the given order) to all but the last of those keys.
  In the latter case, either all keys have weight 0, or $n\le 3$.
\end{theorem} 

\begin{corollary}\label{cor: root stronger}
  For any ordering $\pi$ of the keys,
  there is an optimal tree $T$ such that each subtree $T_N$
  either starts with an inequality comparison,
  or starts with equality tests to all max-weight keys that reach the root $N$ of $T_N$,
  doing the tests in the order in which those keys occur in $\pi$.
\end{corollary}

The corollary follows from the first part of Theorem~\ref{thm: root stronger}.
(Let $T$ be any optimal tree that,
for the subproblem arising at each node,
uses a subtree that satisfies Condition (iii) of the theorem whenever possible,
and otherwise uses a subtree that satisfies Condition (i) or (ii).)

\begin{proof}[Proof of Theorem~\ref{thm: root stronger}]
  Assume that there is an optimal tree whose root does an equality test.
  (Otherwise the theorem holds trivially.)
  Among such trees, let $T$ be one that maximizes the number of equality tests to max-weight keys,
  breaking ties as follows.
  For the full variant, break ties by choosing a tree that maximizes the number of \emph{inversions},
  where an inversion is a pair of nodes such that one is an ancestor of the other,
  with the ancestor doing an inequality test and the descendant doing an equality test to a max-weight key.
  For the succesful-queries variant, break ties by \emph{minimizing} the number of inversions.

  Let $\beta_{\max} = \max_i \beta_ i$. 
  We will use the following utility lemma:
  \begin{lemma}\label{lemma: equal}
    For every maximum-weight key $k_a$,
    along the path traversed by a search for $k_a$ from the root,
    each non-leaf node has side-weight $\beta_{\max}$.
  \end{lemma}
  \begin{proof}
    Let $k_a$ be any max-weight key. 
    The parent of the leaf of $k_a$ does an equality-test to a max-weight key
    (otherwise replacing the comparison at the parent by an equality test to $k_a$
    would preserve optimality and increase the number of equality tests to max-weight keys,
    contradicting the choice of $T$).
    The side-weight of the parent is therefore $\beta_{\max}$.
    By Lemma~\ref{lemma: side weights}, 
    along the path from the parent to the root, the side-weights are non-decreasing.
    On the other hand, the root does an equality test, so its side-weight (which equals its key's weight)
    can't exceed $\beta_{\max}$.
    So every node on the path must have side-weight $\beta_{\max}$.
  \end{proof}

  The body of the proof has four cases.

  \mycase{1} \emph{Suppose that this is a successful-queries instance, 
    and there is an inequality test on the search path to some max-weight key $k_b$.}

  Let $N_b$ be the parent of $k_b$'s leaf, and consider the path from the root to $N_b$.
  The parent $N_b$ does an equality test to a max-weight key
  (otherwise replacing $N_b$ by an equality test to $k_b$
  would give an optimal tree with more equality tests to max-weight keys,
  contradicting the choice of $T$).
  So some inequality test on the path is followed by an equality test.
  Let $\compnode < {k_p}$ and $\compnode = {k_a}$ be such a parent and child.
  Assume that $\compnode = {k_a}$ is the right child of $\compnode < {k_p}$
  (the other case is symmetric).
  As shown in Figure~\ref{fig: case 1}(a),
  let $L$ be the left subtree of $\compnode < {k_p}$
  and let $R$ be the right subtree of $\compnode = {k_a}$.
  By Lemma~\ref{lemma: equal},
  both nodes have side-weight $\beta_{\max}$.
  By definition, this means that $k_a$ has weight $\beta_{\max}$,
  and the minimum of the weights of $\compnode < {k_p}$'s two subtrees is $\beta_{\max}$.

  \begin{figure}\centering \small \renewcommand{\treeLabelSize}{\scriptsize}
    (a)
    \hspace*{-1em}
    \adjustbox{valign=t}{
      \begin{TREE}
        for tree = { s sep=2em }, 
        [{$q < k_p$}
          [{$L$}, subtree]
          [{$q=k_a$}
            [{$k_a$}]
            [{$R$}, subtree]
            ]
          ]
      \end{TREE}
    }
    (b)
    \hspace*{-1em}
    \adjustbox{valign=t}{
      \begin{TREE}
        for tree = { s sep=2em }, 
        [{$q=k_a$}
          [{$k_a$}]
          [{$q < k_p$}
            [{$L$}, subtree]
            [{$R$}, subtree]
            ]
          ]
      \end{TREE}
    }
    (c)
    \hspace*{-1em}
    \adjustbox{valign=t}{
      \begin{TREE}
        for tree = { s sep=2em }, 
        [{$q=k_a$}
          [{$k_a$}]
          [{$L$}, subtree,
            [{$X$}]
            [{$R$}, subtree]
            ]
          ]
      \end{TREE}
    }
    \caption{For Case 1 of the proof of Theorem~\ref{thm: root stronger}.}\label{fig: case 1}
  \end{figure}

  \mycase{1.1}
  \emph{Suppose that the weight of $L$ equals $\beta_{\max}$.}
  Moving the node $\compnode = {k_a}$ just above $\compnode < {k_p}$,
  as shown in Figure~\ref{fig: case 1}(b),
  increases the depth of $L$ by one but reduces the depth of $k_a$'s leaf by one,
  so preserves the cost and yields another optimal tree $T'$.
  But $T'$ has the same number of equality tests to max-weight keys
  (including one at the root),
  and one less inversion than $T$.
  This contradicts the choice of $T$,  so Case 1.1 cannot happen.
  
  \mycase{1.2}
  \emph{Otherwise the weight of $L$ strictly exceeds $\beta_{\max}$.}
  Recalling that the side-weight of $\compnode < {k_p}$ equals $\beta_{\max}$,
  this implies that the weight of $\compnode < {k_p}$'s right subtree equals $\beta_{\max}$,
  which (given that the subtree contains $R$ and the key $k_a$ of weight $\beta_{\max}$)
  implies that all keys in $R$ have weight zero.

  Now consider the following modification to $T$.
  Replace the subtree rooted at $\compnode < {k_p}$ by
  a subtree
  (as shown roughly in Figure~\ref{fig: case 1}(c)) 
  with root $\compnode = {k_a}$,
  whose left child is the leaf of $k_a$,
  and whose right subtree $L'$ is obtained from $L$ as follows.

  Let $Q$ be the set of queries in $\Queries$ whose searches in $T$ reach the root of $R$.
  Because of the node $\compnode < {k_p}$,  each value in $Q$ is at least $k_p$.
  On the other hand, all comparison keys used by nodes in $L$ are strictly less than $k_p$
  (as all queries routed through $L$ are less than $k_p$, and $T$ has no ``redundant'' nodes).
  So no comparison made in $L$ can distinguish between different queries in $Q$.
  So all queries in $Q$, if searched for in $L$,  reach the same leaf, say $X$, in $L$.
  (In the current case, $X$ is the leaf at the end of the right spine.)

  Since this is a successful-queries instance, the leaf $X$ has a single key.
  Obtain $L'$ from $X$ by replacing $X$ by a new equality test to $X$'s key,
  giving the new node left child $X$ and right subtree $R$.
  In the modified tree,
  any query in $Q$ reaches the root of $L$,
  is then routed through $L$ to $R$,
  and then is routed through $R$ just as it was in $T$, to an appropriate leaf.
  Queries not in $Q$ are still handled correctly (although the query to $X$'s leaf has one extra equality test).

  As for cost, the modification decreases the depth of the leaf of $k_a$ by one,
  while increasing the depth $X$'s leaf by one.
  The resulting tree cannot be cheaper than $T$, so $X$'s key must also be a max-weight key.
  But then the modified tree has the same cost but more equality tests to max-weight keys,
  contradicting the choice of $T$.
  So Case 1 cannot happen.

  \mycase{2} \emph{Suppose that this is a full instance, 
    and there is an inequality test on the search path to some max-weight key.}
  Let $\compnode < {k_a}$ be the first inequality-test node on the search path.
  As shown in Figure~\ref{fig: case 2}(a), 
  let $\compnode = {k_p}$ be the parent of $\compnode < {k_a}$,
  let $L$ and $R$ be the subtrees of $\compnode < {k_a}$.
  \begin{figure}
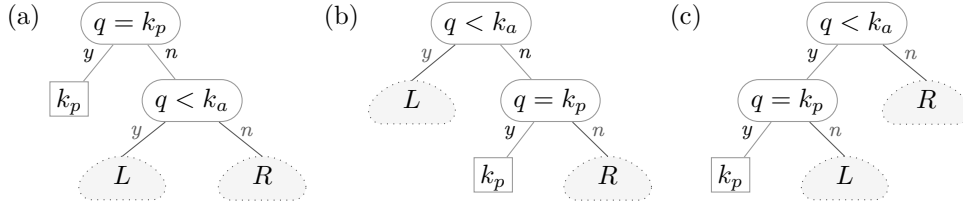
\centering \small \renewcommand{\treeLabelSize}{\scriptsize}
    (a)
    \hspace*{-1em}
    \adjustbox{valign=t}{
      \begin{TREE}
        for tree = { s sep=2em }, 
        [{$q = k_p$}
          [{$k_p$}]
          [{$q<k_a$}
            [{$L$}, subtree]
            [{$R$}, subtree]
            ]
          ]
      \end{TREE}
    }
    (b)
    \hspace*{-1em}
    \adjustbox{valign=t}{
      \begin{TREE}
        for tree = { s sep=2em }, 
        [{$q<k_a$}
          [{$L$}, subtree]
          [{$q = k_p$}
            [{$k_p$}]
            [{$R$}, subtree]
            ]
          ]
      \end{TREE}
    }
    (c)
    \hspace*{-1em}
    \adjustbox{valign=t}{
      \begin{TREE}
        for tree = { s sep=2em }, 
        [{$q<k_a$}
          [{$q=k_p$}
            [{$k_p$}]
            [{$L$}, subtree]
            ]
          [{$R$}, subtree]
          ]
      \end{TREE}
    }
    \caption{For Case 2 of the proof of Theorem~\ref{thm: root stronger}.}\label{fig: case 2}
  \end{figure}
  Consider moving the node $\compnode = {k_p}$ just below $\compnode < {k_a}$,
  either just above $R$ (Figure~\ref{fig: case 2}(b))
  or just above $L$ (Figure~\ref{fig: case 2}(c)).
  One of these modifications preserves correctness
  (the first if $k_p \ge k_a$, otherwise the second).
  Both increase the depth of $k_p$ by one,
  but decrease the depth of $L$ or $R$ by one.
  By Lemma~\ref{lemma: equal},
  the side-weight of each node is $\beta_{\max}$,
  so the weight of $k_p$ is $\beta_{\max}$,
  and the weights of $L$ and $R$ are both at least $\beta_{\max}$.
  So the resulting tree $T'$ is also optimal.
  It has as many equality tests to max-weight keys as $T$,
  and one more inversion, 
  so (by the choice of $T$) the modified tree $T'$ cannot have an equality test at the root.
  So Condition (ii) holds.
  By the assumption that this is a full instance, this suffices to prove the theorem in Case 2.

  \mycase{3}
  \emph{Suppose that $T$ starts by doing equality tests to all the max-weight keys in $\Queries$ in any order.}
  Exchanging any two of the equality-test nodes to the max-weight keys preserves optimality.
  By repeated exchanges, we can modify $T$ to do the equality tests to the max-weight keys in any order.
  So Condition (iii) holds, proving the theorem in Case 3.

  \mycase{4} \emph{In the remaining case, 
    along every search path to every max-weight key all tests are equality tests,
    but $T$ doesn't start by doing equality tests to all the max-weight keys in some order.}
  Let $k_a$ be a max-weight key with maximum leaf depth.
  Let $N_a$ be the parent of $k_a$'s leaf.
  Every node on the path $P_a$ from the root to $N_a$ does an equality test,
  necessarily (by Lemma~\ref{lemma: equal}) to a max-weight key.
  Assume without loss of generality that the equality test at $N_a$ is to $k_a$
  (otherwise take $k_a$ to be the key that $N_a$ does an equality test to, which has the same leaf depth).
  Let $Q$ be the set of queries that don't have equality tests along the path $P_a$.
  By the condition for Case 4, $Q$ contains at least one max-weight key, say $k_b$.
  The search path for each query in $Q$ reaches the right child of $N_a$,
  but (by the choice of $k_a$) the search for $k_b$ must end at that child,
  so that child must be the leaf for $k_b$.
  Hence (for the tree to be correct) $Q$ can contain only $k_b$.
  It follows that $T$ is just a path of equality tests, one for every key except $k_b$.
  It follows that this is a successful-queries instance, and every key has weight $\beta_{\max}$.

  Suppose for contradiction that $\beta_{\max}>0$ and $n\ge 4$.
  Replacing $T$'s equality-test keys appropriately
  yields another optimal tree $T'$ that starts with equality tests to $k_1$, $k_2$, and $k_3$ in that order.
  Let $R$ be the right subtree of $\compnode = {k_3}$, as shown in Figure~\ref{fig: case 4}(a). 
  Let $T''$ be the tree shown in Figure~\ref{fig: case 4}(b).
  The weights of keys $k_1$ and $k_3$ are equal,
  and (as $n\ge 4$) the weight of $R$ is positive,
  implying that $T''$ is strictly cheaper than $T'$.
  This contradicts the optimality of $T'$.
  So $n\le3$ or $\beta_{\max} = 0$, proving the theorem.
  \begin{figure}\centering \small \renewcommand{\treeLabelSize}{\scriptsize}
    (a)
    \hspace*{-1em}
    \adjustbox{valign=t}{
      \begin{TREE}
        for tree = { s sep=2em }, 
        [{$q = k_1$}
          [{$k_1$}]
          [{$q = k_2$}
            [{$k_2$}]
            [{$q = k_3$}
              [{$k_3$}]
              [{$R$}, subtree]
              ]
            ]
          ]
      \end{TREE}
    }
    (b)
    \hspace*{-1em}
    \adjustbox{valign=t}{
      \begin{TREE}
        for tree = { s sep=2em }, 
        [{$q < k_3$}
          [{$q = k_1$}
            [{$k_1$}]
            [{$k_2$}]
            ]
          [{$q = k_3$}
            [{$k_3$}]
            [{$R$}, subtree]
            ]
          ]
      \end{TREE}
    }
    \caption{For Case 4 of the proof of Theorem~\ref{thm: root stronger}.}\label{fig: case 4}
  \end{figure}
\end{proof}

\section{Source code (Python)}\label{sec: code}

See Figure~\ref{fig: code}.


\begin{figure}[!p]
{\tt
\begin{lstlisting}
#!/usr/bin/env python3.9
from math import inf            # infinity
from functools import lru_cache
memoize = lru_cache(maxsize=None)


def opt_cost(K, alpha, beta):
    'Return min cost of any 2-way-comparison tree for (K, alpha, beta).'

    K_by_wt = sorted(range(1, len(K)+1), key=lambda i: beta[i-1])
    
    @memoize
    def opt(i, j, h):

        K_in_A_ijh = [b for b in K_by_wt[:h] if i <= b < j]
        w_ijh = sum(alpha[i:j] + [beta[k-1] for k in K_in_A_ijh])

        return (
            0 if j == i+1 and i not in K_in_A_ijh else
            opt(i, j, h-1) if h > 0 and K_by_wt[h-1] not in K_in_A_ijh else
            w_ijh + min(
                opt(i, j, h-1) if h > 0 else inf,
                min(opt(i, b, h) + opt(b, j, h) for b in range(i+1, j))
                if i+1 < j else inf
            )
        )

    return opt(0, len(K)+1, len(K))


# Run the algorithm on the instance from Figure 1(b).

K = [1, 3, 5, 7]
beta = [0.08, 0.02, 0.08, 0.33]
alpha = [0.15, 0.21, 0.08, 0.04, 0.01]

assert opt_cost(K, alpha, beta) == 2.8
\end{lstlisting}
}
\caption{Python code for the algorithm for the full variant.
  (Technical note: the memoization decorator uses a dictionary.
  To achieve a truly faithful implementation,
  it should be replaced by a three-dimensional array.)
}\label{fig: code}
\end{figure}


\end{document}